\documentclass[a4paper,runningheads,oribibl]{llncs} %

\usepackage[english]{babel}
\usepackage[dvipsnames,usenames]{color}
\usepackage{verbatim,graphics,indentfirst} 
\usepackage{stmaryrd}
\usepackage[nosumlimits]{amsmath}
\usepackage{amssymb}
\usepackage{latexsym}
\usepackage{enumerate}
\usepackage[T1]{fontenc}
\usepackage{hyperref}
\usepackage{mathrsfs}
\usepackage{setspace}
\usepackage[all]{xy}
\usepackage{proof}
\usepackage{xspace} 
\usepackage{cleveref}
\usepackage{bbm}

\usepackage{multicol}

\newcommand{\nats}                  {{\mathbb N}}

\newcommand{\cD}{\mathcal{D}}

\newcommand{\cL}{\mathcal{L}}
\newcommand{\cB}{\mathcal{B}}

\newcommand{\cM}{\mathcal{M}}
\newcommand{\cC}{\mathcal{C}}

\newcommand{\der}    						  {\vdash}

\newcommand{\tuple}[1]                         {{\langle #1\rangle}}

\newcommand{\cT}{\mathcal{T}}
\newcommand{\cV}{\mathcal{V}}
\newcommand{\conn}{{\copyright}}
\newcommand{\conw}{\mathsf{conn}}
\newcommand{\infrule}                              {{\circledR}}

\newcommand{\Mmon}                              {\mathsf{mon}}
\newcommand{\Subst}                              {\mathsf{Sbst}}
\newcommand{\St}                              {\mathsf{S}}
\newcommand{\Mx}                              {\mathsf{M}}
\newcommand{\LogMat}                              {\mathsf{LM}}
\newcommand{\val}{v}

\newcommand{\EqDiv}[3]{#2 \dashv\vdash_{#1} #3}
\newcommand{\TWO}{\mathbbm{2}}

\usepackage{stackengine}
\usepackage{scalerel}

\stackMath
\newcommand\mytilde[1]{\widetilde{#1}}
\def\dyhat{-.5ex}
\newcommand\mynewtilde[1]{\ThisStyle{%
              \stackon[\dyhat]{\SavedStyle#1}
                              {\SavedStyle\widetilde{\phantom{#1}}}}}

\newcommand{\fibS}{\bullet}
\newcommand{\fib}[2]{#1\fibS#2}

\newcommand{\cH}{\mathcal{H}}

\newcommand{\MatV}{\mathsf{V}}
\newcommand{\MatD}{\mathsf{D}}
\newcommand{\MatC}{\mathsf{C}}

\DeclareMathOperator*{\sub}{\mathsf{sbf}}
\DeclareMathOperator*{\var}{\mathsf{var}}

\DeclareMathOperator*{\head}{\mathsf{head}}
\DeclareMathOperator*{\skel}{\mathsf{sk}}

\DeclareMathOperator*{\Rule}{\mathsf{R}}
\DeclareMathOperator*{\imp}{\to} 
\DeclareMathOperator*{\ifelse}{\textsc{if}}
\DeclareMathOperator*{\eq}{\leftrightarrow}

\usepackage{tikz}
\pgfarrowsdeclare{arcs}{arcs}{...} 
{ 
  \pgfsetdash{}{0pt} 
  \pgfsetroundjoin   
  \pgfsetroundcap    
  \pgfpathmoveto{\pgfpoint{-10pt}{10pt}} 
  \pgfpatharc{180}{270}{10pt} 
  \pgfpatharc{90}{180}{10pt} 
  \pgfusepathqstroke 
}
\newdimen\arrowsize 
\pgfarrowsdeclare{arcs'}{arcs'}{...} 
{ 
  \arrowsize=0.5pt 
  \advance\arrowsize by .5\pgflinewidth 
  \pgfsetdash{}{0pt} 
  \pgfsetroundjoin   
  \pgfsetroundcap    
  \pgfpathmoveto{\pgfpoint{-4\arrowsize}{4\arrowsize}} 
  \pgfpatharc{180}{270}{4\arrowsize} 
  \pgfpatharc{90}{180}{4\arrowsize} 
  \pgfusepathqstroke 
}
\newdimen\arrowsize 
\pgfarrowsdeclare{arcs''}{arcs''}{...} 
{ 
  \arrowsize=0.5pt 
  \advance\arrowsize by .5\pgflinewidth 
  \pgfsetdash{}{0pt} 
  \pgfsetroundjoin   
  \pgfsetroundcap    
  \pgfpathmoveto{\pgfpoint{-8\arrowsize}{4\arrowsize}} 
  \pgfpatharc{180}{270}{4\arrowsize} 
  \pgfpatharc{90}{180}{4\arrowsize} 
  \pgfusepathqstroke 
}
\makeatletter
\newbox\xrat@below
\newbox\xrat@above
\newcommand{\coimparrow}[2][]{%
  \setbox\xrat@below=\hbox{\ensuremath{\scriptstyle #1}}%
  \setbox\xrat@above=\hbox{\ensuremath{\scriptstyle #2}}%
  \pgfmathsetlengthmacro{\xrat@len}{max(\wd\xrat@below,\wd\xrat@above)+.65em}%
  \mathrel{\tikz [>-,>=arcs',baseline=-.5ex,line width=.2mm]
                 \draw (0,0) -- node[below=-2pt] {\box\xrat@below}
                                node[above=-2pt] {\box\xrat@above}
                       (\xrat@len,0) ;}}
\makeatother
\newcommand{\coimpl}{\not\to} 
\newcommand{\coimplsub}{\not\to}

\newcommand{\pow}[1]{\mathsf{Pow}(#1)}

\newcommand{\ou}          {\vee}
\newcommand{\e}          {\wedge}

\DeclareMathOperator*{\bbot}{\mathsf{\bot}}
\DeclareMathOperator*{\size}{\mathsf{Size}}

\date{\today}

\title{Merging fragments of classical logic\thanks{This research was done under the scope of R\&D Unit 50008, financed by the applicable financial framework (FCT/MEC
through national funds and when applicable co-funded by FEDER/PT2020), and is part of the MoSH initiative of SQIG at Instituto de Telecomunica\c{c}\~{o}es.
S\'ergio Marcelino acknowledges the FCT
postdoc grant SFRH/BPD/76513/2011.  
Jo\~ao Marcos acknowledges partial support by CNPq and by the Humboldt Foundation.}}
\author{%
Carlos Caleiro\inst{1}
\and
S\'ergio Marcelino\inst{1}
\and 
Jo\~ao Marcos\inst{2}
}

\institute{Departament of Mathematics, IST, Universidade de Lisboa, Portugal\\
          {\tt \{ccal,smarcel\}@math.tecnico.ulisboa.pt}
\and           
           Lo.L.I.T.A.\ and DIMAp, UFRN, Brazil \\
           \email{jmarcos@dimap.ufrn.br}
}

\begin{document}

\maketitle

\begin{abstract}
We investigate the possibility of extending the non-func\-tion\-ally complete logic of a collection of Boolean connectives by the addition of further Boolean connectives that make the resulting set of connectives functionally complete.  
More precisely, we will be interested in checking whether an axiomatization for Classical Propositional Logic may be produced by merging Hilbert-style calculi for two disjoint 
incomplete fragments of it.  We will prove that the answer to that problem is a negative one, unless one of the components includes only top-like connectives.
\end{abstract}

\allowdisplaybreaks

\section{Introduction}
\noindent
Hilbert-style calculi are arguably the most widespread way of defining logics, and simultaneously the least studied one, from the metalogical viewpoint. This is mostly due to the fact that proofs in Hilbert-style calculi are hard to obtain and systematize, in contrast with other proof formalisms such as sequent calculi and their well developed proof-theory, and semantic approaches involving algebraic or relational structures.
Still, Hilbert-style calculi are most directly associated with the fundamental notion of logic as a consequence operation and are thus worth studying.  Furthermore, merging together Hilbert-style calculi for given logics in order to build a combined logic precisely captures the mechanism for combining logics known as \textit{fibring}, yielding the least logic on the joint language that extends the logics given as input~\cite{ccal:car:jfr:css:04}. Fibring fares well with respect to two basic guiding principles one may consider, conservativity and interaction. In contrast, despite their {better behaved compositional character},
alternative approaches based for instance on sequent calculi are prone to emerging interactions and breaches in conservativity (see, for instance, the~\emph{collapsing problem}~\cite{ccal:jabr:05}).

In this paper, as an application of recent results about fibred logics, we investigate the modular construction of Hilbert-style calculi for classical logic. Take, for instance, implication and negation. Together, they form a functionally complete set of connectives. However, all suitable axiomatizations of classical logic we have seen include at least one axiom/rule where implication and negation interact. Rautenberg's general method for axiomatizing fragments of classical logic~\cite{Rautenberg1981}, which explores the structure of Post's lattice~\cite{Post41,Lau:2006:FAF:1205006}, further confirms the intuition about the essential role of interaction axioms/rules, that one may have drawn from any experience with axiomatizations of classical logic. Additionally, such expectation is consistent with a careful analysis of the characterization of the complexity of different fragments of classical logic and their associated satisfiability problems~\cite{REITH20031,Sistla:1985:CPL:3828.3837}, namely in the light of recent results on the decidability and complexity of fibred logics~\cite{deccomp}. The question we wish to give a definitive answer to, here, is precisely this: is it possible to recover classical logic by 
fibring two disjoint fragments of it?
We will show that 
the recovery is successful
iff one of the logics represents a fragment of classical logic consisting only of top-like connectives (i.e., connectives that only produce theorems, for whichever arguments received as input), while the other results in a functionally complete set of connectives with the addition of~$\top$. 

The paper is organized as follows. In \Cref{sec:prelim}, we overview basic notions of logic, including Hilbert calculi and logical matrices, and introduce helpful notation. In \Cref{subsec:fibring} we carefully review the mechanism for fibring logics, as well as some general results about disjoint fibring that shall be necessary next. Our main results, analyzing the merging of disjoint 
fragments of classical logic, are obtained in \Cref{sec:merge}. We conclude, in \Cref{sec:final}, with a brief discussion of further work. To the best of our knowledge, \Cref{prop:conslifts} (\Cref{subsec:fibring}) and all the characterization results in \Cref{sec:merge} are new.

\section{Preliminaries}
\label{sec:prelim}


\subsection{Logics in abstract}
\label{subsec:syntax}

\noindent
In what follows, a \textsl{signature}~$\Sigma$ is an indexed set $\{\Sigma^{(k)}\}_{k\in\mathbb{N}}$, where each~$\Sigma^{(k)}$ is a collection of $k$-place \textsl{connectives}.  
Given a signature~$\Sigma$ and a (disjoint) set~$P$ of \textsl{sentential variables}, we denote by $L_\Sigma(P)$ the absolutely free $\Sigma$-algebra generated by~$P$, also known as \textsl{the language generated by~$P$ over~$\Sigma$}.  The objects in~$L_\Sigma(P)$ are called \textsl{formulas}, and a formula is called \textsl{compound} in case it belongs to $L_\Sigma(P)\setminus P$, that is, in case it contains some connective.  
We will sometimes use $\head(C)$ to refer to the main connective in a compound formula~$C$, and say that a formula~$C$ is \textsl{$\Sigma$-headed} if $\head(C)\in\Sigma$. 
Furthermore, we will use $\sub(C)$ to refer to the set of subformulas of~$C$, and use $\var(C)$ to refer to the set of sentential variables occurring in~$C$; the definitions of $\sub$ and $\var$ are extended to sets of formulas in the obvious way.
Given a formula~$C$ such that $\var(C)\subseteq\{p_1,\ldots,p_k\}$, it is sometimes convenient to take it as inducing a \textsl{$k$-ary term function} $\varphi=\lambda p_1\ldots p_k.C$ such that $\varphi(p_1,\ldots,p_k)=C$, 
over which we will employ essentially the same terminology used to talk about connectives and formulas therewith constructed ---in particular, a $k$-ary term function is 
induced by a formula generated by~$k$ distinct sentential variables over a $k$-place connective.
In such cases we will also say that the corresponding term functions are \textsl{allowed by} the underlying language and \textsl{expressed by} the corresponding logic.
We will often employ the appellations \textsl{nullary} for $0$-ary and \textsl{singulary} for $1$-ary term functions (or for the connectives that induce them).
Given signatures $\Sigma\subseteq\Sigma^\prime$ and sets $P\subseteq P^\prime$ of sentential variables, a \textsl{substitution} is a structure-preserving mapping over the corresponding sets of formulas, namely a function $\sigma:P\longrightarrow L_{\Sigma^\prime}(P^\prime)$ which 
extends uniquely to a homomorphism $\sigma^\star:L_\Sigma(P)\longrightarrow L_{\Sigma^\prime}(P^\prime)$ by setting $\sigma^\star(\conn(C_1,\ldots,C_k)):=\conn(\sigma^\star(C_1),\ldots,\sigma^\star(C_k))$ for every $\conn\in\Sigma^{(k)}$. We shall refer to~$\sigma^\star(C)$ more simply as~$C^\sigma$. The latter notation is extended in the natural way to sets of formulas: given $\Pi\subseteq L_{\Sigma}(P)$, $\Pi^\sigma$ denotes $\{C^\sigma:C\in\Pi\}$.

A \textsl{logic}~$\cL$ over the language~$L_\Sigma(P)$ is here a structure $\tuple{L_\Sigma(P),\der}$ equipped with a so-called \textsl{consequence relation}~$\der\;\subseteq\pow{L_\Sigma(P)}\times L_\Sigma(P)$ respecting (\textbf{R}) $\Gamma\cup\{C\}\der C$; 
  (\textbf{M}) if $\Gamma\der C$ then $\Gamma\cup\Delta\der C$; (\textbf{T}) if $\Gamma\der D$ for every $D\in\Delta$ and $\Gamma\cup\Delta\der C$, then $\Gamma\der C$; and (\textbf{SI}) if $\Gamma\der C$ then $\Gamma^\sigma\der C^\sigma$ for any substitution $\sigma:P\longrightarrow L_{\Sigma}(P)$. Any assertion in the form $\Pi\der E$ will be called a \textsl{consecution}, and may be read as `$E$ follows from~$\Pi$ (according to~$\cL$)'; whenever $\tuple{\Pi,E}\in\;\der$ one may say that \textsl{$\cL$ sanctions $\Pi\der E$}.  Henceforth, union operations and braces will be omitted from consecutions, and the reader will be trusted to appropriately supply them in order to make the expressions well-typed.  

Given two logics $\cL=\tuple{L_\Sigma(P),\der}$ and $\cL^\prime=\tuple{L_{\Sigma^\prime}(P^\prime),\der^\prime}$, 
we say that \textsl{$\cL^\prime$ extends~$\cL$} in case $P\subseteq P^\prime$, $\Sigma\subseteq\Sigma^\prime$ and $\der\;\subseteq\;\der^\prime$.  In case $\Gamma\der C$ iff $\Gamma\der^\prime C$, for every $\Gamma\cup\{C\}\subseteq L_\Sigma(P)$, we say that the extension is \textsl{conservative}.  So, in a conservative extension no new consecutions are added in the `reduced language' $L_\Sigma(P)$ by the `bigger' logic~$\cL^\prime$ to those sanctioned by the `smaller' logic~$\cL$.
%
Fixed $\cL=\tuple{L_\Sigma(P),\der}$, and given $\Sigma\subseteq\Sigma^\prime$ and $P\subseteq P^\prime$, 
let $\Subst$ collect all the substitutions $\sigma:P\longrightarrow L_{\Sigma^\prime}(P^\prime)$.
We say that a formula~$B$ of $L_{\Sigma^\prime}(P^\prime)$ is a \textsl{substitution instance} of a formula~$A$ of $L_{\Sigma}(P)$ if there is a substitution $\sigma\in\Subst$ such that $A^\sigma=B$.
A \textsl{natural conservative extension induced by~$\cL$} is given by the logic $\cL^\prime=\tuple{L_{\Sigma^\prime}(P^\prime),\der^\prime}$ equipped by the smallest sub\-sti\-tu\-tion-in\-var\-i\-ant consequence relation preserving the consecutions of~$\cL$ inside the extended language, that is, such that $\Gamma\der^\prime C$ iff there is some $\Delta\cup\{D\}\subseteq L_{\Sigma}(P)$ and some $\sigma\in\Subst$ such that $\Delta\der D$, where $\Delta^{\!\sigma}=\Gamma$ and $D^\sigma=C$.  In what follows, when we simply enrich the signature and the set of sentential variables, we shall not distinguish between a given logic and its natural conservative extension.

Two formulas~$C$ and~$D$ of a logic $\cL=\tuple{L_\Sigma(P),\der}$ are said to be \textsl{logically equivalent according to} $\cL$ if $C\der D$ and $D\der C$; two sets of formulas~$\Gamma$ and~$\Delta$ are said to be logically equivalent according to~$\cL$ if each formula from each one of these sets may correctly be said to follow from the other set of formulas (notation: $\EqDiv{\cL}{\Gamma}{\Delta}$).
We call the set of formulas $\Gamma\subseteq L_\Sigma(P)$ \textsl{trivial} (\textsl{according to~$\cL$}) if $\EqDiv{\cL}{\Gamma}{L_\Sigma(P)}$.
We will say that the logic~$\cL$ is \textsl{consistent} if its consequence relation~$\der$ does not sanction all possible consecutions over a given language, that is, if there is some set of formulas $\Pi\cup\{E\}$ such that $\Pi\not\der E$, in other words, if~$\cL$ contains some non-trivial set of formulas~$\Pi$; we call a logic \textsl{inconsistent} if it fails to be consistent.  
We say that a set of formulas~$\Pi$ in $\cL=\tuple{L_\Sigma(P),\der}$ is \textsl{$\der$-explosive} in case $\Pi^\sigma\der E$ for every substitution~$\sigma:P\longrightarrow L_{\Sigma}(P)$ and every formula~$E$.  Obviously, an inconsistent logic $\cL$ is one in which the empty set of formulas is $\der$-explosive.

Fixed a denumerable set of sentential variables~$P$ and a non-empty signature~$\Sigma$, let $\conw=\bigcup\Sigma$.  To simplify notation, whenever the context eliminates any risk of ambiguity, we will sometimes refer to $L_\Sigma(P)$ more simply as~$L_{\conw}$.  For instance, given the $2$-place connective~$\land$, in writing $L_\land$ we refer to the language generated by~$P$ using solely the connective~$\land$, and similarly for the $2$-place connective~$\lor$ and the language~$L_\lor$.  Taking the union of the corresponding signatures, in writing $L_{\land\lor}$ we refer to the \textsl{mixed language} whose formulas may be built using exclusively the connectives~$\land$ and~$\lor$.

\begin{example}\label{abstractCL}
For an illustration involving some familiar connectives, 
a logic $\cL=\tuple{L_\Sigma(P),\der}$ will be said to be \textsl{$\conn$-classical} if, for every set of formulas $\Gamma\cup\{A,B,C\}$ in its language (see, for instance,~\cite{Humberstone2011-HUMTC}): 
\smallskip

\noindent
\scalebox{0.95}{
\begin{tabular}{ll}
  {[$\conn=\top\in\Sigma^{(0)}$]} & $\Gamma,\top\der C$ implies $\Gamma\der C$ \\
  {[$\conn=\bot\in\Sigma^{(0)}$]} & $\Gamma\der \bot$ implies $\Gamma\der C$ \\ 
  {[$\conn=\neg\in\Sigma^{(1)}$]} & (i) $A,\neg A\der C$;  and 
  (ii) $\Gamma,A\der C$ and $\Gamma,\neg A\der C$ imply $\Gamma\der C$ \\
  {[$\conn=\land\in\Sigma^{(2)}$]} & $\Gamma,A\land B\der C$ iff $\Gamma,A,B\der C$ \\
  {[$\conn=\lor\in\Sigma^{(2)}$]} & $\Gamma,A\lor B\der C$ iff $\Gamma,A\der C$ and $\Gamma,B\der C$ \\
  {[$\conn=\imp\in\Sigma^{(2)}$]} & (i) $A,A\imp B\der B$; 
  (ii) $\Gamma,A\imp B\der C$ implies $\Gamma,B\der C$; \\
  & and (iii) $\Gamma,A\der C$ and $\Gamma,A\imp B\der C$ implies $\Gamma\der C$\\
\end{tabular}}\\[1mm]
\noindent
Other classical connectives may also be given appropriate abstract characterizations, `upon demand'.  
%
%
If the logic $\cL_{\conw}=\tuple{L_{\conw},\der}$ is $\conn$-classical for every $\conn\in \conw$, we call it \textsl{the logic of classical~$\conw$} and denote  it by $\cB_{\conw}$.
\hfill$\triangle$
\end{example}

Let~$\varphi$ be some $k$-ary term function expressed by the logic $\cL=\tuple{L_\Sigma(P),\der}$.  
If $\varphi(p_1,\ldots,p_k)\der p_j$ for some $1\leq j\leq k$, we say that~\textsl{$\varphi$ is projective over its $j$-th component}.  Such term function is called a \textsl{projection-conjunction} if it is logically equivalent to its set of projective components, i.e., if there is some $J\subseteq\{1,2,\ldots,k\}$ such that (i) $\varphi(p_1,\ldots,p_k)\der p_j$ for every $j\in J$ and (ii) $\{p_j:j\in J\}\der \varphi(p_1,\ldots,p_k)$.  
In case $\varphi(p_1,\ldots,p_k)\der p_{k+1}$, we say that~$\varphi$ is \textsl{bottom-like}.
We will call~$\varphi$ \textsl{top-like} if $\der \varphi(p_1,\ldots,p_k)$; do note that the latter is a particular case of projection-conjunction (take $J=\varnothing$).
Classical conjunction is another particular case of projection-conjunction (take $n=2$ and $J=\{1,2\}$); its singulary version (take $n=1$ and $J=\{1\}$) corresponds to the so-called \textsl{affirmation connective}.
A term function that is neither top-like nor bottom-like will here be called \textsl{significant}; 
if in addition it is not a projection-conjunction, we will call it \textsl{very significant}; in each case, connectives shall inherit the corresponding terminology from the term functions that they induce. Note that being not very significant means being either bottom-like or a projection-conjunction.

%

%
%
%


\subsection{Hilbert-style proof systems}
\label{subsec:hilbert}

\noindent
One of the standard ways of presenting a logic is through the so-called `axiomatic approach'.  We call \textsl{Hilbert calculus} over the language $L_\Sigma(P)$ any structure $\cH=\tuple{L_\Sigma(P),\Rule}$, \textsl{presented by} a set of \textsl{inference rules} $\Rule\subseteq\pow{L_\Sigma(P)}\times L_\Sigma(P)$.  An inference rule $\infrule=\tuple{\Delta,D}\in\Rule$ is said to have \textsl{premises}~$\Delta$ and \textsl{conclusion}~$D$, and is often represented in tree-format by writing~$\frac{\Delta}{D}{{}^{{}_\infrule}}$, 
or $\frac{D_{{}^1} \; \ldots\; D_{n}}{D}{{}^{{}_\infrule}}$ when $\Delta=\{D_1,\ldots,D_n\}$, or $\frac{}{D}{{}^{{}_\infrule}}$ in case $\Delta=\varnothing$.  The latter type of rule, with an empty set of premises, is called \textsl{axiom}.

Fix in what follows a Hilbert calculus presentation~$\cH=\tuple{L_\Sigma(P),\Rule}$, and consider signatures $\Sigma\subseteq\Sigma^\prime$ and sets $P\subseteq P^\prime$ of sentential variables, with the corresponding collection $\Subst$ of substitutions from $L_{\Sigma}(P)$ into $L_{\Sigma^\prime}(P^\prime)$.  
Given formulas $\Gamma\cup\{C\}\subseteq L_{\Sigma^\prime}(P^\prime)$, a \textsl{rule application allowing to infer~$C$ from~$\Gamma$ according to~$\cH$} corresponds to a pair $\tuple{\infrule,\sigma}$ such that $\frac{\Delta}{D}{{}^{{}_\infrule}}$ is in~$\Rule$ and $\sigma\in\Subst$, while $\Delta^{\!\sigma}=\Gamma$ and $D^\sigma=C$.
Such rule applications are often annotated with the names of the corresponding rules being applied.
In case $\Delta=\varnothing$ we may also refer to the corresponding rule application as an \textsl{instance of an axiom}.
As usual, an \textsl{$\cH$-derivation of~$C$ from~$\Gamma$} is a tree~$\cT$ with the following features: 
(i)~all nodes are labelled with substitution instances of formulas of $L_{\Sigma}(P)$; 
(ii)~the root is labelled with~$C$; 
(iii)~the existing leaves are all labelled with formulas from~$\Gamma$; 
(iv)~all non-leaf nodes are labelled with instances of axioms, or with premises from~$\Gamma$, or with formulas inferred by rule applications from the formulas labelling the roots of certain subtrees of~$\cT$, using the inference rules~$\Rule$ of~$\cH$.
It is not hard to see that~$\cH$ induces a logic $\cL_\cH=\tuple{L_{\Sigma^\prime}(P^\prime),\der_{\Rule}}$ by setting $\Gamma\der_{\Rule} C$ iff there is some $\cH$-derivation of~$C$ from~$\Gamma$; indeed, we may safely leave to the reader the task of verifying that postulates (\textbf{R}), (\textbf{M}), (\textbf{T}) and (\textbf{SI}) are all respected by~$\der_{\Rule}$.  
We shall say that a logic $\cL=\tuple{L_\Sigma(P),\der}$ is \textsl{characterized by a Hilbert calculus~$\cH=\tuple{L_\Sigma(P),\Rule}$} iff $\der\;=\;\der_{\Rule}$.

\begin{example}\label{hilbertCL}
We revisit the well-known connectives of classical logic whose inferential behaviors were 
described in \Cref{abstractCL}. What follows are the rules of appropriate Hilbert calculi for the logics $\cL_\conn=\tuple{L_\conn,\der_{\Rule_\conn}}$, where $p,q,r\in P$:
\smallskip

\noindent
\scalebox{0.95}{
\begin{tabular}{ll}
  {[$\conn=\top$]} & $\frac{}{\;\;\top\;\;}{{}^{{}_{{}_{\mathsf{t}1}}}}$\\[2mm]
  {[$\conn=\bot$]} & $\frac{\;\;\bot\;\;}{p}{{}^{{}_{{}_{\mathsf{b}1}}}}$\\[2mm]
  {[$\conn=\neg$]} & 
    $\frac{p}{\;\;\neg\neg p\;\;}{{}^{{}_{{}_{\mathsf{n}1}}}}\quad
    \frac{\;\;\neg \neg p\;\;}{p}{{}^{{}_{{}_{\mathsf{n}2}}}}\quad 
    \frac{\;\;p\qquad  \neg p\;\;}{q}{{}^{{}_{{}_{\mathsf{n}3}}}}$\\[2mm]
  {[$\conn=\land$]} & 
    $\frac{\;p\e q\;}{\;\; p\;\;}{{}^{{}_{{}_{\mathsf{c}1}}}}\quad
    \frac{\;p\e q\;}{\;\; q\;\;}{{}^{{}_{{}_{\mathsf{c}2}}}}\quad
    \frac{\;\;p\qquad  q\;\;}{p\e q}{{}^{{}_{{}_{\mathsf{c}3}}}}$\\[2mm]
  {[$\conn=\lor$]} & 
    $\frac{p}{p\ou q}{{}^{{}_{{}_{\mathsf{d}1}}}} \quad
    \frac{p\ou p}{p}{{}^{{}_{{}_{\mathsf{d}2}}}} \quad 
    \frac{p\ou q}{q\ou p}{{}^{{}_{{}_{\mathsf{d}3}}}} \quad
    \frac{p\ou (q \ou r)}{(p \ou q)\ou r}{{}^{{}_{{}_{\mathsf{d}4}}}}$\\[2mm]
  {[$\conn=\imp$]} & 
  $\frac{}{p\imp (q\imp p)}{{}^{{}_{{}_{\mathsf{i}1}}}}\hspace{2mm}
  \frac{}{(p\imp(q\imp r))\imp ((p\imp  q)\imp (p\imp  r))}{{}^{{}_{{}_{\mathsf{i}2}}}}\hspace{2mm}
  \frac{}{((p\imp  q)\imp  p)\imp  p}{{}^{{}_{{}_{\mathsf{i}3}}}}\hspace{2mm}
  \frac{\;p\quad p\imp  q\;}{q}{{}^{{}_{{}_{\mathsf{i}4}}}}$\\[2mm]
 \end{tabular}}\\[1mm]
 Of course, other classical connectives can also be axiomatized. For instance, the bi-implication~$\eq$ defined by the term function $\lambda pq.(p\imp q)\land(q\imp p)$ may be presented by:
\\[1mm]
\noindent
{%
\scalebox{0.95}{\begin{tabular}{ll} 
   {[$\conn=\eq$]} & 
  $\frac{}{(p\eq(q\eq r))\eq ((p\eq q)\eq r)}{{}^{{}_{{}_{\mathsf{e}1}}}}\quad
  \frac{}{((p\eq r)\eq(q\eq p))\eq(r\eq q)}{{}^{{}_{{}_{\mathsf{e}2}}}}\quad 
  \frac{\;p\quad p\eq  q\;}{q}{{}^{{}_{{}_{\mathsf{e}3}}}}$ 

\end{tabular}}
}
\hfill$\triangle$
\end{example}

\subsection{Matrix semantics}
\label{subsec:semantics-D}

\noindent
Another standard way of presenting a logic is through `model-theoretic semantics'.  
A \textsl{matrix semantics}~$\cM$ over the language $L_\Sigma(P)$ is a collection of logical matrices over $L_\Sigma(P)$, where by a \textsl{logical matrix}~$\LogMat$ over $L_\Sigma(P)$ we mean a structure $\LogMat=\tuple{\MatV,\MatD,\MatC}$ in which the set~$\MatV$ is said to contain \textsl{truth-values}, each truth-value in $\MatD\subseteq\MatV$ is called \textsl{designated}, and for each $\conn\in\Sigma^{(k)}$ there is in~$\MatC$ a $k$-ary \textsl{interpretation} mapping $\mytilde{\conn}$ 
over~$\MatV$.  
A \textsl{valuation} over a logical matrix~$\LogMat$ is any mapping $\val:L_\Sigma(P)\longrightarrow\MatV$ such that $\val(\conn(C_1,\ldots,C_k))=\mytilde{\conn}(\val(C_1),\ldots,\val(C_k))$ for every $\conn\in\Sigma^{(k)}$. We denote by $\cV_\LogMat$ the set of all valuations over $\LogMat$, and say that the valuation~$\val$ over $\LogMat$ \textsl{satisfies} a formula $C\in L_\Sigma(P)$ if $\val(C)\in\MatD$.
Note that a valuation 
might be thought more simply as a mapping $\val:P\longrightarrow\MatV$, given that 
there is a unique extension of~$\val$ as a homomorphism from $L_\Sigma(P)$ into the similar algebra having~$\MatV$ as carrier and having each symbol $\conn\in\Sigma^{(k)}$ interpreted as the $k$-ary operator $\mytilde{\conn}:\MatV^k\longrightarrow\MatV$.
Analogously, each $k$-ary term function $\lambda p_1\ldots p_k.\varphi$ over $L_\Sigma(P)$ is interpreted by a logical matrix $\LogMat$ in the natural way as a $k$-ary operator $\mytilde{\varphi}:\MatV^k\longrightarrow\MatV$.
We shall call $\mathcal{\cC}_\LogMat^\Sigma$ the collection of all term functions compositionally derived over~$\Sigma$ and interpreted through $\LogMat$; in the literature on Universal Algebra, $\mathcal{\cC}_\LogMat^\Sigma$ is known as the \textsl{clone} of operations definable by term functions allowed by the signature~$\Sigma$, under the interpretation provided by~$\LogMat$.

Given a valuation $\val:L_\Sigma(P)\longrightarrow\MatV$, where the truth-values $\MatD\subseteq\MatV$ are taken as designated, and given formulas $\Gamma\cup\{C\}\subseteq L_\Sigma(P)$, we say that~$C$ \textsl{follows from}~$\Gamma$ \textsl{according to}~$\val$ (notation: $\Gamma\der_\val C$) iff it is not the case that~$\val$ simultaneously satisfies all formulas in~$\Gamma$ while failing to satisfy~$C$.  
We extend the definition to a set~$\cV$ of valuations by setting $\Gamma\der_{\cV}C$ iff $\Gamma\der_{\val}C$ for every $\val\in\cV$, that is, $\der_{\cV}\;=\bigcap_{\val\in\cV}(\der_\val)$.
%
On its turn, a matrix semantics~$\cM$ defines a consequence relation~$\der_\cM$ by setting $\Gamma\der_{\cM}C$ iff $\Gamma\der_{\cV_\LogMat}C$ for every $\LogMat\in\cM$, that is, $\der_{\cM}\;=\;\bigcap_{\LogMat\in\cM}(\der_{\cV_\LogMat})$. If we set $\cV_\cM:=\bigcup_{\LogMat\in\cM}({\cV_\LogMat})$, it should be clear that $\der_{\cM}\;=\;\der_{\cV_\cM}$.
We shall say that a logic $\cL=\tuple{L_\Sigma(P),\der}$ is \textsl{characterized by a matrix semantics~$\cM$} iff $\der\;=\;\der_{\cM}$.
{To make precise what we mean herefrom by a `fragment' of a given logic,}
given a subsignature $\Sigma^\prime\subseteq\Sigma$, a \textsl{sublogic}~$\cL^\prime$ of~$\cL$ is a logic $\cL^\prime=\tuple{L_{\Sigma^\prime}(P),\der^\prime}$ characterized by a matrix semantics~$\cM^\prime$ such that the interpretation~$\mytilde{\conn}$ of the connective~$\conn$ is the same at both~$\cM$ and~$\cM^\prime$, for every $\conn\in\Sigma^\prime$ and every $\LogMat\in\cM^\prime$.  It is not hard to see that $\cL$ will in this case consist in a conservative extension of~$\cL^\prime$.
There are well-known results in the literature to the effect that any logic whose consequence relation satisfies $(\mathbf{R})$, $(\mathbf{M})$, $(\mathbf{T})$ and $(\mathbf{SI})$ may be characterized by a 
matrix semantics~\cite{Wojcicki88}.

\begin{example}\label{MatSem1}
We now revisit yet again the connectives of classical logic that received our attention at \Cref{abstractCL,hilbertCL}.
Let $\MatV=\{0,1\}$ and $\MatD=\{1\}$.  
Given a 
logical matrix $\tuple{\MatV,\MatD,\MatC}$, we will call it \textsl{$\conn$-Boolean} if:
\smallskip

\noindent
\scalebox{0.95}{%
\begin{tabular}{ll}
  {[$\conn=\top$]} & $\mytilde{\top}=1$\\[1mm]
  {[$\conn=\bot$]} & $\mytilde{\bot}=0$\\[1mm]
  {[$\conn=\neg$]}
    & (i) $\mytilde{\neg}(1)=0$; and (ii) $\mytilde{\neg}(0)=1$\\[1mm]
  {[$\conn=\land$]}  
    & (i) $\mytilde{\land}(1,1)=1$; and (ii) $\mytilde{\land}(x,y)=0$ otherwise\\[1mm]
  {[$\conn=\lor$]}  
    & (i) $\mytilde{\lor}(0,0)=0$; and (ii) $\mytilde{\lor}(x,y)=1$ otherwise\\
\end{tabular}\\[1mm]
}
\scalebox{0.95}{%
\begin{tabular}{ll}
  {[$\conn=\imp$]}  
    & (i) $\mytilde{\imp}(1,0)=0$; and (ii) $\mytilde{\imp}(x,y)=1$ otherwise\\
\end{tabular}\\[1mm]
}
\scalebox{0.95}{%
\begin{tabular}{ll}
  {[$\conn=\eq$]}  
    & (i) $\mytilde{\eq}(x,y)=1$ if $x=y$; and (ii) $\mytilde{\eq}(x,y)=0$ otherwise\\[1mm]
\end{tabular}}\\[1mm]
\noindent It is not difficult to show that, if~$\cM$ is a collection of $\conn$-Boolean logical matrices, 
the logic $\cL_\conn=\tuple{L_\conn,\der_{\cM}}$ is $\conn$-classical.
Conversely, every $\conn$-classical logic may be characterized by a single $\conn$-Boolean logical matrix. 

We take the chance to introduce a few other connectives that will be useful later on.  These connectives may be primitive in some sublogics of classical logic, but can also be 
defined by term functions involving the previously mentioned connectives, as follows:
{\noindent\center{
\scalebox{0.95}{
\begin{tabular}{rll}
$\coimpl$ & $:=$ & $\lambda p q.\neg(p\imp q)$\\
$+$ & $:=$ & $\lambda p q.(p\land \neg q)\lor(q\land \neg p)$\\
$\ifelse$ & $:=$ & $\lambda p q r.(p\imp q)\land(\neg p\imp r)$\\
$T^n_0$ & $:=$& $\lambda p_1\dots p_n.\top$, for $n\geq 0$\\ 
$T^n_n$ & $:=$& $\lambda p_1\dots p_n.p_1\land \dots\land p_n$, for $n>0$\\ 
$T^n_k$ & $:=$& $\lambda p_1\dots p_n.(p_1\land T^{n-1}_{k-1}(p_2,\dots,p_n))\lor T^{n-1}_{k}(p_2,\dots,p_n)$, for $n>k>0$\\ 
\end{tabular}}\\[1mm]
}}

\noindent
Note that a logical matrix containing such connectives is $\conn$-Boolean if:\smallskip

\noindent
\scalebox{0.95}{\begin{tabular}{ll}
{[$\conn=\,\coimpl$]}  
    & (i) $\mynewtilde{\coimpl}(1,0)=1$; and (ii) $\mynewtilde{\coimpl}(x,y)=0$ otherwise\\[1mm]
    {[$\conn=+$]}  
    & (i) $\mytilde{+}(x,y)=0$ if $x=y$; and (ii) $\mytilde{+}(x,y)=1$ otherwise\\[1mm]
%

{[$\conn=\ifelse$]}  
    & (i) $\mytilde{\ifelse}(1,y,z)=y$; and (ii) $\mytilde{\ifelse}(0,y,z)=z$\\[1mm]

{[$\conn=\mytilde{T_k^n}$]}  
    & (i) $\widetilde{T_k^n}(x_1,\ldots,x_n)=0$ if $\size(\{i:x_i=1\})<k$; \\
    & and (ii) $\widetilde{T_k^n}(x_1,\ldots,x_n)=1$ otherwise
\end{tabular}}\\[-6mm]
%
%

\hfill$\triangle$
\end{example}

In what follows we shall use the expression \textsl{two-valued logic} to refer to any 
logic 
characterized by the logical matrix
$\{\MatV_\TWO,\MatD_\TWO,\MatC\}$, where $\MatV_\TWO=\{0,1\}$ and $\MatD_\TWO=\{1\}$, and use the expression \textsl{Boolean connectives} to refer to the corresponding $2$-valued interpretation of the symbols in~$\Sigma$ (see \Cref{MatSem1}).  From this perspective, whenever we deal with a two-valued logic whose language is expressive enough, modulo its interpretation through a 
matrix semantics, to allow for all operators of a Boolean algebra $\mathsf{BA}$ over~$\MatV_\TWO$ to be compositionally derived, we will say that we are dealing with \textsl{classical logic}.  
Alternatively, whenever the underlying signature turns out to be of lesser importance, one might say that classical logic is the two-valued logic that corresponds to the clone $\cC_{\mathsf{BA}}$ 
containing all operations over~$\MatV_\TWO$.  Due to such level of expressiveness, classical logic is said thus to be \textsl{functionally complete} (\textsl{over~$\MatV_\TWO$}). 
On those grounds, it follows that all two-valued logics may be said to be sublogics of classical logic.  
The paper~\cite{Rautenberg1981} shows how to provide a Hilbert calculus presentation for any proper two-valued sublogic of classical logic. %
%

Emil Post's characterization of functional completeness for classical logic~\cite{Post41,Lau:2006:FAF:1205006} is very informative. First of all, it tells us that there are exactly five maximal functionally incomplete clones (i.e, co-atoms in Post's lattice), namely:
{\noindent\center{
\scalebox{0.95}{\begin{tabular}{llllllllllllllllllllllllll}
$\mathbb{P}_0$ & $=$ & $\cC_{\mathsf{BA}}^{\lor\,\coimplsub}$ &\qquad\qquad&
$\mathbb{P}_1$ & $=$ & $\cC_{\mathsf{BA}}^{\land\imp}$ &\qquad\qquad&
$\mathbb{A}$ & $=$ & $\cC_{\mathsf{BA}}^{\eq\bot}$ & \qquad\qquad&
$\mathbb{M}$ & $=$ & $\cC_{\mathsf{BA}}^{\land\lor\top\bot}$& \qquad\qquad&
$\mathbb{D}$ & $=$ & $\cC_{\mathsf{BA}}^{T^3_2 \neg}$
\end{tabular}}\\[2mm]
}}
\noindent
The Boolean top-like connectives form the clone $\mathbb{UP}_1=\cC_{\mathsf{BA}}^{\top}$. 
As it will be useful later on, we mention that an analysis of Post's lattice also reveals that there are also a number of clones which are maximal with respect to~$\top$, i.e., functionally incomplete clones that become functionally complete by the mere addition of the nullary connective~$\top$ (or actually any other connective from $\mathbb{UP}_1$). 
In terms of the Post's lattice, the clones whose join with $\mathbb{UP}_1$ result in $\cC_{\mathsf{BA}}$ are:
{\noindent\center{
\scalebox{0.95}{\begin{tabular}{llllllllllllllllllllllllll}
$\mathbb{D}$&\qquad\qquad\qquad&
$\mathbb{T}^\infty_0$ & $=$ & $\cC_{\mathsf{BA}}^{\coimplsub}$ &\qquad\qquad\qquad&
$\mathbb{T}^n_0$ & $=$ & $\cC_{\mathsf{BA}}^{T^{n+1}_n\coimplsub}$ (for $n\in\nats$) 
\end{tabular}}\\[2mm]
}}

\noindent
It is worth noting that $\mathbb{T}^1_0=\mathbb{P}_0$. 

If a logic turns out to be 
characterized by a single logical matrix with a finite set of truth-values, a `tabular' decision procedure is associable to its consequence relation based on the fact that the valuations over a finite number of sentential variables may be divided into a finite number of equivalence classes, and one may then simply do an exhaustive check for satisfaction whenever a finite number of formulas is involved in a given consecution.
More generally, we will say that a logic~$\cL$ is \textsl{locally tabular} if the relation of logical equivalence $\EqDiv{\cL}{}{}$ partitions the language $L_\Sigma(\{p_1,\ldots,p_k\})$, freely generated by the signature $\Sigma$ over a finite set of sentential variables, into a finite number of equivalence classes. 
It is clear that all two-valued sublogics of classical logic are locally tabular.
On the same line, it should be equally clear that any logic that fails to be locally tabular 
cannot be characterized by a 
logical matrix with a finite set of truth-values. 


\section{Combining logics}
\label{subsec:fibring}

\noindent
Given two logics $\cL_a=\tuple{L_{\Sigma_a\!}(P),\der_a}$ and $\cL_b=\tuple{L_{\Sigma_b}(P),\der_b}$, their \textsl{fibring} is defined as the smallest logic $\fib{\cL_a}{\cL_b}=\tuple{L_{a\fibS b}(P),\der_{a\fibS b}}$, where $L_{a\fibS b}(P)=L_{\Sigma_a\cup \Sigma_b}(P)$, and where $\der_a\;\subseteq\;\der_{a\fibS b}$ and $\der_b\;\subseteq\;\der_{a\fibS b}$, that is, it consists in the smallest logic over the joint signature that extends both logics given as input.  Typically, one could expect the combined logic $\fib{\cL_a}{\cL_b}$ to \textit{conservatively} extend both~$\cL_a$ and~$\cL_b$.  That is not always possible, though (consider for instance the combination of a consistent logic with an inconsistent logic).  
A full characterization of the combinations of logics through disjoint fibring that yield conservative extensions of both input logics may be found at~\cite{apconsfinval}.
The fibring of two logics is called \textsl{disjoint} (or \textsl{unconstrained}) if their signatures are disjoint.
A neat characterization of fibring is given by way of Hilbert calculi: Given $\der_a\;=\;\der_{\Rule_a}$ and $\der_b\;=\;\der_{\Rule_b}$, where $\Rule_a$ and $\Rule_b$ are sets of inference rules, we may set $\Rule_{a\fibS b}:=\Rule_a\cup\Rule_b$ and then note that $\fib{\cL_a}{\cL_b}=\tuple{L_{a\fibS b}(P),\der_{\Rule_{a\fibS b}}}$.

Insofar as a logic may be said to codify inferential practices used in reasoning, the (conservative) combination of two logics should not only allow one to faithfully recover the original forms of reasoning sanctioned by each ingredient logic over the respective underlying language, but should also allow the same forms of reasoning ---and no more--- to obtain over the mixed language.
Hence, it is natural to think that each of the ingredient logics cannot see past the connectives belonging to the other ingredient logic ---the latter connectives look like `monoliths' whose internal structure is inaccessible from the outside.

To put things more formally, given signatures $\Sigma\subseteq\Sigma^\prime$ and given a formula $C\in L_{\Sigma^\prime}(P)$, we call \textsl{$\Sigma$-monoliths} the largest subformulas of~$C$ whose heads belong to $\Sigma^\prime\setminus\Sigma$.  Accordingly, the set $\Mmon_{\Sigma}(C)\subseteq \sub(C)$ of all $\Sigma$-monoliths of~$C$ is defined by setting:

\mbox{
\centering{
\scalebox{.95}{\parbox{10cm}{
\begin{align*}
	\Mmon_{\Sigma}(C) & := 
	\begin{cases}
		\varnothing & \mbox{ if } C\in P,\\
		\bigcup_{i=1}^k\Mmon_{\Sigma}(C_i) & \mbox{ if } C = \conn(C_1,\ldots,C_k) \mbox{ and } \conn\in \Sigma^{(k)},\\
		\{C\} & \mbox{ otherwise.}
	\end{cases}
\end{align*}
}}}}

\noindent
This definition may be extended to sets of formulas in the usual way, by setting $\Mmon_{\Sigma}(\Gamma):=\bigcup_{C\in\Gamma}\Mmon_{\Sigma}(C)$.  Note, in particular, that $\Mmon_{\Sigma}(\Gamma)=\varnothing$ if $\Gamma\subseteq L_{\Sigma}(P)$.  
{From the viewpoint of the signature~$\Sigma$, monoliths may be seen as `skeletal' (sentential) variables that represent formulas of $L_{\Sigma^\prime}(P)$ whose inner structure cannot be taken advantage of.  In what follows, let $X^{\Sigma^\prime\!}:=\{x_D:D\in L_{\Sigma^\prime}(P)\}$ be a set of fresh symbols for sentential variables.  Given $C\in L_{\Sigma^\prime}(P)$, in order to represent the \textsl{$\Sigma$-skeleton of~$C$} we  define the function $\skel_\Sigma:L_{\Sigma^\prime}(P)\longrightarrow L_{\Sigma^\prime}(P\cup X^{\Sigma^\prime})$ by setting:}

\mbox{
\centering{
\scalebox{.95}{\parbox{10cm}{
\begin{align*}
	\skel\nolimits_{\Sigma}(C)&:= 
	\begin{cases}
		C & \mbox{ if } C\in P,\\
		\conn(\skel\nolimits_{\Sigma}(C_1),\ldots,\skel\nolimits_{\Sigma}(C_k))\hspace{-3mm}& \mbox{ if } C = \conn(C_1,\ldots,C_k) \mbox{ and } \conn\in \Sigma^{(k)},\\
		x_{C},& \mbox{ otherwise.}
	\end{cases}
\end{align*}
}}}}

\noindent
Clearly, a skeletal variable~$x_D$ is only really useful in case $\head(D)\in\Sigma^\prime\setminus\Sigma$.

\begin{example}\label{fib-ex1}
Recall from \Cref{hilbertCL} the inference rules characterizing the logic~$\cB_\land$ of classical conjunction and the logic~$\cB_\lor$ of classical disjunction.  As in \Cref{abstractCL}, we let $\cB_{\land\lor}$ refer to a logic that is at once $\land$-classical and $\lor$-classical, and contains no other primitive connectives besides~$\land$ and~$\lor$.
Consider now the fibred logic $\cL_{\fib{\land}{\lor}}:=\fib{\cB_\land}{\cB_\lor}$.  It should be clear that $\der_{\land\fibS\lor}\;\subseteq\;\der_{\land\lor}$.  It is easy to see now that $\EqDiv{\land\fibS\lor}{p\land(p\lor q)}{p}$ (a logical realization of an absorption law of lattice theory).  Indeed, a one-step derivation~$\cD_1$ of~$p$ from $p\land(p\lor q)$ in $\cL_{\land\fibS\lor}$ is obtained simply by an application of rule~$\mathsf{c}1$ to~$p\land(p\lor q)$, and a two-step derivation~$\cD_2$ of $p\land(p\lor q)$ from~$p$ in $\cL_{\land\fibS\lor}$ is obtained by the application of rule $\mathsf{d}1$ to $p$ to obtain $p\lor q$, followed by an application of $\mathsf{c}3$ to~$p$ and $p\lor q$ to obtain $p\land(p\lor q)$.
Note that $\Mmon_{\Sigma_\land}(p\land(p\lor q))=\{p\lor q\}$ and $\Mmon_{\Sigma_\lor}(p\land(p\lor q))=\{p\land(p\lor q)\}$, and note also that $\skel_{\Sigma_\land}(p\land(p\lor q))=p\land x_{p\lor q}$ and $\skel_{\Sigma_\lor}(p\land(p\lor q))=x_{p\land (p\lor q)}$.  This means that from the viewpoint of~$\cB_\land$ the step of~$\cD_2$ in which the foreign rule $\mathsf{d}1$ is used is seen as a `mysterious' passage from~$p$ to a new sentential variable $x_{p\lor q}$ taken \textit{ex nihilo} as an extra hypothesis in the derivation, and from the viewpoint of~$\cB_\lor$ the step of~$\cD_2$ in which the foreign rule $\mathsf{c}3$ is used is seen as the spontaneous introduction of an extra hypothesis $x_{p\land(p\lor q)}$.
At our next example we will however show that the dual absorption law, represented by $\EqDiv{\land\fibS\lor}{p\lor(p\land q)}{p}$, does not hold, even though the corresponding equivalence holds good over all Boolean algebras.  This will prove that $\der_{\land\lor}\;\not\subseteq\;\der_{\land\fibS\lor}$, and thus $\cB_{\land\lor}\not\subseteq\fib{\cB_\land}{\cB_\lor}$.
\hfill$\triangle$
\end{example}

\begin{remark}\label{cor:seqmon}
In a natural conservative extension, where the syntax of a logic is extended with new connectives but no further inference power is added, it is clear that formulas headed by the newly added connectives are treated as monoliths. Hence, the following result from~\cite{deccomp} applies: Given 
$\cL=\tuple{L_\Sigma(P),\der}$, $\Sigma\subseteq\Sigma'$ and $\Delta\cup \{C,D\} \subseteq L_{\Sigma'}(P)$ we have
$\Delta \der D\mbox{ if and only if }\skel\nolimits_{\Sigma}(\Delta)\der \skel\nolimits_{\Sigma}(D).$
\hfill$\triangle$
\end{remark}

We will present next a fundamental result from~\cite{deccomp} that fully describes disjoint mixed reasoning in~$\fib{\cL_{a}}{\cL_{b}}$, viz.\ by identifying the consecutions sanctioned by such combined logic with the help of appropriate consecutions sanctioned by~its ingredient logics~$\cL_{a}$ and~$\cL_{b}$.  Given that consecutions in~$\der_{a\fibS b}$ are justified by alternations of consecutions sanctioned by~$\der_a$ and consecutions sanctioned by~$\der_b$, given a set of mixed formulas $\Delta\subseteq L_{a\fibS b}$, we define the \textsl{saturation $\St_{a\fibS b}(\Delta)$ of}~$\Delta$ as $\bigcup_{n\in\mathbb{N}}\St_{a\fibS b}^n(\Delta)$, where $\St_{a\fibS b}^0(\Delta):=\Delta$ and $\St_{a\fibS b}^{n+1}(\Delta):=\{D\in\sub(\Delta):\St_{a\fibS b}^{n}(\Delta)\der_a D\mbox{ or }\St_{a\fibS b}^{n}(\Delta)\der_b D\}$.  
In addition, given a set of mixed formulas $\Delta\cup\{D\}\subseteq L_{a\fibS b}$, we abbreviate by $\Mx^i_{a\fibS b}(\Delta,D)$ the set of $\Sigma_i$-monoliths $\{C\in\Mmon_{\Sigma_i}(D):\Delta\der_{a\fibS b}C\}$, for each $i\in\{a,b\}$.  Such ancillary notation helps us stating: 

\begin{theorem}\label{bazooka}
Let $\cL_{a}$ and $\cL_{b}$ be two logics, each one characterizable by a single logical matrix.
If $\cL_{a}$ and $\cL_{b}$ have disjoint signatures, the consecutions in the fibred logic $\cL_{a\fibS b}$ are such that  
$\Gamma\der_{a\fibS b}C$ iff the following condition holds good:\smallskip\\
$(\mathbf{Z}^a)$\ \ $\St_{a\fibS b}(\Gamma),\Mx^a_{a\fibS b}(\Gamma,C)\der_a C$\ \  or\ \ $\St_{a\fibS b}(\Gamma)$ is $\der_b$-explosive.
\end{theorem}

\noindent
Note that the roles of $a$ and $b$ may be exchanged in the above theorem, given that the fibring operation is obviously commutative, so we might talk accordingly of a corresponding condition $(\mathbf{Z}^b)$, in case it turns out to be more convenient.  The original formulation of this result in~\cite{deccomp} was based on a slightly more sophisticated notion of saturation, which reduces to the above one in particular when the logics involved in the combination are characterizable by means of a truth-functional semantics (i.e., a matrix semantics involving a single logical matrix), as it is indeed the case for all sublogics of classical logic.

\begin{example}\label{fib-ex2}
Set $a=\land$ and $b=\lor$, $E=p\lor(p\land q)$, and let $\Gamma=\{E\}$ and $C=p$.
Note that (i) $\sub(\Gamma)=\{p,q,p\land q,p\lor(p\land q)\}$.
Moreover, it is clear that (ii) $\Mmon_{\Sigma_a}(p)=\Mmon_{\Sigma_b}(p)=\varnothing$, given that $p\in P$, thus $\Mx^a_{a\fibS b}(\Gamma,C)=\Mx^b_{a\fibS b}(\Gamma,C)=\varnothing$.  We know by the base case of the definition of~$\St$ that (iii) $\St_{a\fibS b}^0(\Gamma)=\Gamma$.  Let us now show that $\St_{a\fibS b}^1(\Gamma)=\Gamma$, from which it follows that $\St_{a\fibS b}(\Gamma)=\Gamma$.  We shall be freely making use of item~$(\mathsf{a})$ of \Cref{cor:seqmon}.  Note first, by (\textbf{R}), that we obviously have $\Gamma\der_c E$, for $c\in\{a,b\}$, and note also that (iv) $\skel_{\Sigma_a}(E)=x_{E}$, (v) $\skel_{\Sigma_b}(E)=p\lor x_{p\land q}$, (vi) $\skel_{\Sigma_a}(p\land q)=p\land q$, (vii) $\skel_{\Sigma_b}(p\land q)=x_{p\land q}$ and (viii) $\skel_{\Sigma_c}(r)=r$ when $r\in\{p,q\}$, for $c\in\{a,b\}$.  To see that $\St_{a\fibS b}^0(\Gamma)\not\der_c D$ for every $D\in\sub(\Gamma)\setminus\{E\}$ in case~$c$ is~$a$ it suffices to invoke (i), (iii), (iv), (vi) and (viii), and set a valuation~$v$ such that $\val(x_E):=1$ and $\val(p)=\val(q):=0$; in case~$c$ is~$b$ it suffices to invoke (i), (iii), (v), (vii) and (viii), and one may even reuse the previous valuation~$v$, just adding the extra requirement that $\val(x_{p\land q}):=0$.  It thus follows from the recursive case of the definition of~$\St$ that $\St_{a\fibS b}^1(\Gamma)=\Gamma$.  It is easy to see, with the help of (iv) and (v), that $\St_{a\fibS b}(\Gamma)=\{E\}$ is neither $\der_a$-explosive nor $\der_b$-explosive.  Therefore, according to condition $(\mathbf{Z}^c)$ in \Cref{bazooka}, to check whether $\Gamma\der_{a\fibS b}C$ one may in this case simply check whether $\Gamma\der_a C$ or $\Gamma\der_b C$.  From the preceding argument about $\St_{a\fibS b}^0(\Gamma)$ we already know that the answer is negative in both cases.  We conclude that $p\lor(p\land q)\not\der_{a\fibS b}p$, thus indeed the fragment of classical logic with conjunction and disjunction as sole primitive connectives must be a non-conservative extension of the fibring of the logic of classical conjunction with the logic of classical disjunction, as we had announced at the end of \Cref{fib-ex1}.
\hfill$\triangle$
\end{example}



The following is the first useful new result of this paper, establishing that conservativity is preserved by disjoint fibring, here proved for the (slightly simpler) case where each logic is characterized by a single logical matrix.

\begin{proposition}\label{prop:conslifts}
Let $\cL_a$ and $\cL_b$ be logics with disjoint signatures, each characterizable by means of a single logical matrix.
If $\cL_a$ and $\cL_b$  conservatively extend logics $\cL_1$ and $\cL_2$, respectively, then $\fib{\cL_a}{\cL_b}$ also conservatively extends $\fib{\cL_1}{\cL_2}$.
\end{proposition}
\begin{proof}
{%
Let $\Sigma_{a}$, $\Sigma_{b}$,  $\Sigma_{1}$ and $\Sigma_{2}$,
 be the signatures of, respectively, $\cL_a$, $\cL_b$, $\cL_1$ and~$\cL_2$.
Fix $\Gamma\cup\{C\}\subseteq L_{\Sigma_1\cup \Sigma_2}(P)$.
From \Cref{bazooka} we may conclude that: (a)~$\Gamma\der_{\fib{a}{b}} C$
  if and only if
  either $\St_{\fib{a}{b}}(\Gamma),\Mx^{a}_{\fib{a}{b}}(\Gamma,C)\der_{a} C$, or $\St_{\fib{a}{b}}(\Gamma)$ is $\der_{b}$-explosive; 
  (b) $\Gamma\der_{\fib{1}{2}} C$
  if and only if
  either $\St_{\fib{1}{2}}(\Gamma),\Mx^1_{\fib{1}{2}}(\Gamma,C)\der_{1} C$, or $\St_{\fib{1}{2}}(\Gamma)$ is $\der_{2}$-explosive.
  Now, from the fact that 
   $\St_{\fib{a}{b}}^n(\Gamma)\cup\St_{\fib{1}{2}}^n(\Gamma)\subseteq L_{\Sigma_1\cup \Sigma_2}(P)$, for all $n\in \nats$, together with
    the assumptions that $\cL_a$ conservatively extends~$\cL_1$ and~$\cL_b$ 
conservatively extends~$\cL_2$ 
we conclude 
   that $\St_{\fib{a}{b}}(\Gamma)=\St_{\fib{1}{2}}(\Gamma)$.
The assumption about conservative extension also guarantees that
(c) $\St_{\fib{a}{b}}(\Gamma)$ is $\der_{b}$-explosive if and only if
$\St_{\fib{1}{2}}(\Gamma)$ is $\der_{2}$-explosive.

  We prove, by induction on the structure of~$C$, that (d) $\Gamma\der_{\fib{a}{b}} C$ if and only if  $\Gamma\der_{\fib{1}{2}} C$.
  If $C$ is a sentential variable then 
  $\Mx^{a}_{\fib{a}{b}}(\Gamma,C)\subseteq \Mmon_{\Sigma_{a}}(C)=\varnothing$
  and, also,
   $\Mx^{1}_{\fib{1}{2}}(\Gamma,C)\subseteq \Mmon_{\Sigma_{1}}(C)=\varnothing$.
   We note that (d) then follows from (a), (b) and (c).
  For the induction step, let $C$ be compound.
   From the inductive hypothesis 
   we conclude that $\Mx^{a}_{\fib{a}{b}}(\Gamma,C)=\Mx^{1}_{\fib{1}{2}}(\Gamma,C)$. 
   Hence, again from (a), (b) and (c), we note that (d) follows.%
}%
\qed\end{proof}

 \section{Merging fragments}\label{sec:merge}

\noindent
This section studies the expressivity of logics obtained by fibring disjoint fragments of classical logic.
We start by analyzing the cases in which combining disjoint 
sublogics
of classical logic still yields a 
sublogic 
of classical logic.

\begin{proposition}\label{lem:clonestop} 
{%
Let
 $\conn_1$ be a Boolean connective and $\conn_2$ be top-like. %
We then have that  
$\fib{\cB_{\conn_1}\!}{\cB_{\conn_2}}=\cB_{\conn_1\conn_2}$.
}%
\end{proposition}
\begin{proof}
{By assumption, $\conn_2$ is top-like,
   hence: 
   ($\star$) for any given set of formulas~$\Delta$, we have $\Delta\der_{\conn_2} \psi$ iff $\psi\in \Delta$ or $\head(\psi)=\conn_2$.
  Let us prove that $\Gamma\der_{\conn_1\conn_2}\varphi$ iff $\Gamma\der_{\conn_1\fibS \conn_2}\varphi$. 
  By \Cref{bazooka}, we know that $\Gamma\der_{\conn_1\fibS \conn_2}\varphi$ iff 
  $\St_{\conn_1\fibS \conn_2}(\Gamma),\Mx^{\conn_1}_{\conn_1\fibS \conn_2}(\Gamma,\varphi)\der_{\conn_1}\varphi$ or $\St_{\conn_1\fibS \conn_2}(\Gamma)$ is $\der_{\conn_2}$-explosive.
  By ($\star$) it follows that if $\St_{\conn_1\fibS \conn_2}(\Gamma)$ is $\der_{\conn_2}$-explosive then $\St_{\conn_1\fibS \conn_2}(\Gamma)$ 
  must contain all the sentential variables and $\{\conn_1\}$-headed formulas. 
   Furthermore, 
 $\Mmon_{\conn_2}(\sub(\Gamma))\subseteq \St_{\conn_1\fibS \conn_2}(\Gamma)$ and
$\Mx^{\conn_1}_{\conn_1\fibS \conn_2}(\Gamma,\varphi)=\Mmon_{\conn_2}(\varphi)$.
  Therefore, $\Gamma\der_{\conn_1\fibS \conn_2}\varphi$ iff 
  $\St_{\conn_1\fibS \conn_2}(\Gamma),\Mx^{\conn_1}_{\conn_1\fibS \conn_2}(\Gamma,\varphi)\der_{\conn_1}\varphi$.
Moreover, $\St_{\conn_1\fibS \conn_2}(\Gamma)=
\{\psi\in \sub(\Gamma):\Gamma,\Mmon_{\conn_1}(\Gamma)\der_{\conn_1}\psi\}$.
We may then finally conclude that $\Gamma\der_{\conn_1\fibS \conn_2}\varphi$ iff $\Gamma,\Mmon_{\conn_1}(\Gamma\cup \{\varphi\})\der_{\conn_1}\varphi$
iff $\Gamma\der_{\conn_1\conn_2}\varphi$.
}%
\qed\end{proof}

\begin{example}
$\fib{\cB_{\,\coimplsub}}{\cB_{\top}}=\cB_{\,\coimplsub \top}$ yields full classical logic, as the set $\{\coimpl,\top\}$ is functionally complete.
\hfill$\triangle$
\end{example}


\begin{proposition}\label{prop:projbot} 
Let $\conn_1$ and $\conn_2$ be Boolean connectives neither of which are very significant. 
 Then, 
$\fib{\cB_{\conn_1\!}}{\cB_{\conn_2}}=\cB_{\conn_1\conn_2}$.
\end{proposition}
\begin{proof}
There are three possible combinations, either \textbf{(a)} both connectives are conjunction-projections, or \textbf{(b)} both are bottom-like, or \textbf{(c)} one connective is bottom-like and the other is a~conjunction-projection.

\noindent
  [Case \textbf{(a)}] Let $J_1$ and $J_2$ be the sets of indices corresponding respectively to the projective components of~$\conn_1$ and of~$\conn_2$.
 For each $\psi\in L_{\conn_1\conn_2}(P)$ let us define 
  $P_\psi\subseteq P$ recursively, in the following way: $P_\psi:=\{\psi\}$ if $\psi\in P$ and
   $P_{\conn_i(\psi_1,\ldots,\psi_k)}:=\bigcup_{a\in J_i}P_{\psi_a}$ for $i\in \{1,2\}$.
We claim that~$\psi$ is equivalent to~$P_\psi$ both according to $\fib{\cB_{\conn_1\!}}{\cB_{\conn_2}}$ and according to $\cB_{\conn_1\conn_2}$. 
 Let us prove this by induction on the structure of~$\psi$.
 For the base case, let~$\psi$ be a sentential variable, and note that~$\psi$ is equivalent to itself.
 If~$\psi$ is a nullary connective~$\conn_i$, for some $i\in\{1,2\}$ (and therefore $\conn_i$ is top-like), then $\conn_i$ is equivalent to $P_{\conn_i}$ (namely, the empty set).
 For the inductive step, consider $\psi=\conn_i(\psi_1,\ldots,\psi_{k_i})$ where $k_i$ is the arity of $\conn_i$. 
 Using the fact that $\conn_i$ is a projection-conjunction we have that $\conn_i$ is equivalent to $\{\psi_a:a\in J_i\}$.
 By induction hypothesis, each $\psi_a$ is equivalent to $P_{\psi_a}$, hence~
 $\psi$ is equivalent to $\bigcup_{a\in J_i}P_{\psi_a}$.
Finally, for a set of sentential variables $B\cup \{b\}$ we clearly have that
  $B\der_{\conn_1\fibS \conn_2}b$ iff $B\der_{\conn_1\conn_2}b$ iff $b\in B$. 
  So, the logics are equal.%
  \smallskip
 
 

\noindent
  [Case \textbf{(b)}] This is similar to the previous case. Let $\psi\in L_{\conn_1\conn_2}(P)$. We now define $A_\psi$ recursively in the following way: $A_\psi:=\{\psi\}$ if $\psi\in P$ or $\head(\psi)=\conn_2$, and
   $A_{\conn_1(\psi_1,\ldots,\psi_k)}:=\bigcup_{a\in J_1}A_{\psi_a}$.
Again, it is not hard to check that in both $\fib{\cB_{\conn_1\!}}{\cB_{\conn_2}}$ and $\cB_{\conn_1\conn_2}$ we have that 
 $\psi$ is equivalent to $A_\psi$. 
Moreover, given $B\cup \{b\}\subseteq P\cup \{\psi: \head(\psi)=\conn_2\}$ we clearly have that
  $B\der_{\conn_1\fibS \conn_2}b$ iff $B\der_{\conn_1\conn_2}b$ iff $b\in B$ or there is $\psi\in B$ such that $\head(\psi)=\conn_2$.
  \smallskip
  
\noindent
  [Case \textbf{(c)}] It should be clear that
    according to both $\fib{\cB_{\conn_1\!}}{\cB_{\conn_2}}$ and $\cB_{\conn_1\conn_2}$ we may conclude that~$\varphi$ follows from~$\Gamma$ iff either $\varphi\in \Gamma$ or there is $\psi\in \Gamma$ such that $\psi\notin P$.  
\qed\end{proof}

{
\begin{proposition}\label{lem:cloneseq} 
For any set of Boolean connectives
 $\conw\subseteq\cC_{\mathsf{BA}}^{\eq} $, we have that $\fib{\cB_{\conw}}{\cB_{\bbot}}= \cB_{\conw\cup\{\bot\}}$.
\end{proposition}
\begin{proof}
We first show that $\fib{\cB_{\eq}}{\cB_{\bbot}}=\cB_{\eq\bot}$.
As $\cB_{\bot}$ is axiomatized by just the single rule $\frac{\bot}{p}$, it
 easily follows that (a) $\Gamma \der_{\eq\bullet \bot} C$ iff $\Gamma \der_{\eq} \bot$ or $\Gamma \der_{\eq} C$.
By \cite[Exercise 7.31.3(iii)]{Humberstone2011-HUMTC}, we note that (b) for every $\Gamma\cup\{B,C\}\subseteq L_{\eq}(P)$ we have that $\Gamma,B\der_{\eq} C$   iff
  $\Gamma\der_{\eq} C$ or $\Gamma\der_{\eq} B \eq C$. Note in addition that (c) $ \der_{\eq} B\eq ((B \eq A) \eq A)$.
Now, if $\Gamma \not\der_{\eq\bullet \bot} A$ then by (a) we have that
$\Gamma \not\der_{\eq} A$ and $\Gamma \not\der_{\eq} \bot$.
Further, using (b) and (c), it follows also that $\Gamma,A\eq\bot \not\der_{\eq} A$ and $\Gamma,A\eq\bot  \not\der_{\eq} \bot$.
Now, a straightforward use of the Lindenbaum-Asser lemma shows that there exists a $\der_{\eq}$-theory $T$ extending $\Gamma\cup\{A\eq\bot\}$ which is maximal relative to $A$.
Obviously $\bot\notin T$, and $\fib{\cB_{\eq}}{\cB_{\bbot}}=\cB_{\eq\bot}$ then follows from the completeness of the axiomatization of $\cB_{\eq}$.
%
%
%
From this, given $\conw\subseteq\cC_{\mathsf{BA}}^{\eq}$, 
we conclude with the help of \Cref{prop:conslifts} that $\fib{\cB_{\conw}}{\cB_{\bbot}}=\cB_{\conw\bot}$.
%
    \qed
 \end{proof}
}

\begin{example}
{%
For every connective~$\conn$ expressed by the logic of classical bi-impli\-ca\-tion, e.g. $\conn\in \{\eq,\lambda pqr.p+q+r\}$,  we have that $\fib{\cB_{\conn}}{\cB_{\bbot}}=\cB_{\conn\bot}$.
}%
\hfill$\triangle$
\end{example}


{%
We now analyze the cases in which combining disjoint sublogics of classical logic results in a logic strictly weaker than the logic of the corresponding classical mixed language.
}

\begin{remark}\label{theconnectives}
{%
A detailed analysis of Post's lattice tells us 
}%
that every clone $\cC_{\mathsf{BA}}^\Sigma$ that contains the Boolean function of a very significant connective (i.e., $\cC_{\mathsf{BA}}^\Sigma\not\subseteq \cC_{\mathsf{BA}}^{\land\top\bot}$) must contain the Boolean function associated to at least one of the following connectives: $\neg$, $\imp$, $\eq$, $\coimpl$, $+$, $\textsc{if}$, $T^{n+1}_n$ (for $n\in\nats$), $T^{n+1}_2$ (for $n\in\nats$), $\lambda pqr.p\lor(q\land r)$, $\lambda pqr.p\lor(q+r)$, $\lambda pqr.p\land(q\lor r)$, $\lambda pqr.p\land(q\imp r)$, $\lambda pqr.p+q+r$.
\hfill$\triangle$
\end{remark}

\begin{lemma}\label{lem:clonesbot} 
 Let $\conw$ be a family of Boolean connectives, 
 and assume that $\cB_{\conw}$ expresses at least one among the  
connectives in \Cref{theconnectives}, distinct from~$\eq$ and~$\lambda p q r. p+q+r$.
 Then $\fib{\cB_{\conw}}{\cB_{\bbot}}\subsetneq \cB_{\conw\cup\{\bot\}}$.
\end{lemma}
\begin{proof}


Let $\conn$ be one of the above 
Boolean connectives.
We show that there are $\Gamma\cup\{C\}\subseteq L_{\conn}(P)$ and $\sigma:P\longrightarrow P\cup \{\top\}$ such that
$\Gamma^\sigma\der_{\conn\bot} C^\sigma$ yet $\Gamma^\sigma\not\der_{\conn\fibS\bot} C^\sigma$,
thus concluding that $\fib{\cB_{\conn}}{\cB_{\bbot}}\subsetneq \cB_{\conn\bot}$.
Hence, by applying \Cref{prop:conslifts}, we obtain that  $\fib{\cB_{\conw}}{\cB_{\bbot}}\subsetneq \cB_{\conw\cup\{\bot\}}$ for $\conw$ in the conditions of the statement.

  
We will explain two cases in detail, and for the remaining cases we just present the relevant formulas~$\Gamma^\sigma$ and~$C^\sigma$, as the rest of the reasoning is analogous.\\[.2mm]
{[Case $\conn=\neg$]} Set $\Gamma:=\varnothing$ and $C^\sigma:=\neg \bot$.
 We have that 
$\der_{\conn\bot}\neg \bot$.
However, since $\not\der_{\conn}\neg(x_\bot)$ and $\St_{\fib{\conn}{\bot}}(\Gamma)=\varnothing$ is not $\der_{\bot}$-explosive,
we conclude that $\not\der_{\fib{\conn}{\bot}}\neg(\bot)$ by \Cref{bazooka}.\\[.2mm]
{[Case $\conn=\ou$]} Set $\Gamma^\sigma:=\{\bot \ou q\}$ and $C^\sigma:=q$.
We have that $\bot \ou q\der_{\conn\bot} q$.
However, since $x_\bot \ou q \not \der_{\conn} q$ and 
$\St_{\fib{\conn}{\bot}}(\{\varphi(x_\bot,q)\})=\{\varphi(x_\bot,q)\}$ is not $\der_{\bot}$-explosive,
we conclude that $\bot \ou q\not\der_{\fib{\conn}{\bot}}q$ by \Cref{bazooka}.\\[.2mm]
%
%
{[Case $\conn=+$]} Set $\Gamma^\sigma:=\{\bot + q\}$ and $C^\sigma:=q$.\\[.2mm]
%
{[Case $\conn=\imp$]} Set $\Gamma^\sigma:=\varnothing$ and $C^\sigma:=\bot \imp q$. \\[.2mm]
%
{[Case $\conn=\ \coimpl$]} let $\Gamma^\sigma:=\{p\}$ and $C^\sigma:=p\coimpl \bot$. \\[.2mm]
%
{[Case $\conn=\lambda p q r. p{\ou} (q {+} r)$}] Set $\Gamma^\sigma:=\{\bot{\ou} (q {+} \bot)\}$ and $C^\sigma:=q$. \\[.2mm]
%
{[Case $\conn=\lambda p q r. p{\e} (q {\imp} r)$]} Set $\Gamma^\sigma:=\{p\}$ and $C^\sigma:=p{\e} (\bot {\imp} r)$. \\[.2mm]
%
{[Case $\conn=\lambda p q r. p{\e} (q {\ou} r)$]} Set $\Gamma^\sigma:=\{p{\e} (\bot {\ou} r)\}$ and $C^\sigma:=r$. \\[.2mm]
%
{[Case $\conn=\lambda p q r. p{\ou} (q{\e} r)$]} Set $\Gamma^\sigma:=\{\bot{\ou} (q {\e} r)\}$ and $C^\sigma:=q$. \\[.2mm]
{[Case $\conn=\ifelse$]} Set $\Gamma^\sigma:=\{\ifelse(\bot,q,r)\}$ and $C^\sigma:=r$. \\[.4mm]
%
{[Case $\conn=T_k^{k+1}$]} Set $\Gamma^\sigma:=\{T_k^{k+1}(p,\ldots,p,q,\bot)\}$ and $C^\sigma:=q$. \\[.2mm]
%
{[Case $\conn=T_2^{k+1}$]} Set $\Gamma^\sigma:=\{T_2^{k+1}(p,p,\bot,\ldots,\bot)\}$ and $C^\sigma:=p$. 
\qed
%
%
%
%
%
\end{proof}

{ 
\begin{corollary}\label{pro:connbot} 
 Let 
 $\conn\notin \cC_{\mathsf{BA}}^{\eq}$ be some very significant Boolean connective.
 %
Then, $\fib{\cB_{\conn}}{\cB_{\bbot}}\subsetneq \cB_{\conn\bot}$.
\end{corollary}
\begin{proof}
  Note, by \Cref{theconnectives} and the fact that both~$\eq$ and $\lambda p q r. p+q+r$ belong to $\cC_{\mathsf{BA}}^{\eq}$, that $\conn$ fulfills the conditions of application of \Cref{lem:clonesbot}. 
 \qed\end{proof}}

\begin{example}\label{ex:sub1}
 For every connective $\conn$ among  
 $\neg$, $\imp$, 
 $\coimpl$, $+$, $\textsc{if}$, $T^{n+1}_n$ (for $n\in\nats$), $T^{n+1}_2$ (for $n\in\nats$), $\lambda pqr.p\lor(q\land r)$, $\lambda pqr.p\lor(q+r)$, $\lambda pqr.p\land(q\lor r)$, and $\lambda pqr.p\land(q\imp r)$,  
 we have that $\fib{\cB_{\conn}}{\cB_{\bbot}}\subsetneq \cB_{\conn\bot}$.
\hfill$\triangle$
\end{example} 

\begin{remark}\label{twovalsignif}
On a two-valued logic: 
(i) sentential variables are always significant, every nullary connective is either top-like or bottom-like; 
(ii) top-like term functions are always assigned the value~$1$ and bottom-like term functions are always assigned the value~$0$; 
(iii) significant singulary term functions all behave semantically either as Boolean affirmation or as Boolean negation. 
\hfill$\triangle$
\end{remark} 

\begin{lemma}\label{lem:TFalphas}
The logic of a significant
Boolean $k$-place connective~$\conn$ expresses some $1$-ary significant compound term function. 
\end{lemma}

\begin{proof}
Let $\varphi$ denote the singulary term function induced by the formula~$\conn(\overline{p})$ obtained by substituting a fixed sentential variable~$p$ at all argument positions of $\conn(p_1,\ldots,p_k)$. If~$\varphi$ is significant, we are done.  Otherwise, there are two cases to consider. 

For the first case, suppose that $\varphi$ is top-like. Thus, given that~$
\conn$ is significant and the logic is two-valued, we know from \Cref{twovalsignif}(ii), in particular, that there must be some valuation~$v$ such that $\val(\conn(p_1,\ldots,p_k))=0$.  Set $I:=\{i:\val(p_i)=1\}$, and define the substitution~$\sigma$ by 
$\sigma(p_j):=\varphi(p)$ if $j\in I$, and $\sigma(p_j):=p$ otherwise.
Let $\psi$ denote the new singulary term function induced by $(\conn(p_1,\ldots,p_k))^\sigma$.  On the one hand, choosing a valuation~$\val^\prime$ such that $\val^\prime(p)=0$ we may immediately conclude that~$\val^\prime(\psi(p))=\val(\conn(p_1,\ldots,p_k))=0$.
On the other hand, choosing~$\val^{\prime\prime}$ such that $\val^{\prime\prime}(p)=1$ we see that
$\val^{\prime\prime}(\sigma(p_j))=1$ for every $1\leq j\leq k$. 
We conclude $\val^{\prime\prime}(\psi(p))=\val^{\prime\prime}(\conn(\overline{p}))=\val^{\prime\prime}(\varphi(p))$, thus $\val^{\prime\prime}(\psi(p))=1$, for~$\varphi$ was supposed in the present case to be top-like.
It follows that~$\psi(p)$ is indeed equivalent here to the sentential variable~$p$.

For the remaining case, where we suppose that~$\varphi$ is bottom-like, it suffices to set $I:=\{i:\val(p_i)=0\}$ and then reason analogously.  In both the latter cases our task is seen to have been accomplished in view of \Cref{twovalsignif}(i). 
\qed
\end{proof}

%

\begin{lemma}\label{lem:projsubst}
Let $\cL=\tuple{L_\Sigma(P),\der}$ be a two-valued logic whose language allows a very significant $k$-ary term function~$\varphi$, 
let~$I$ be the set of indices that identify the projective components of~$\varphi$, 
and let~$\sigma$ be some substitution such that $\sigma(p_i)=p_i$, for $i\in I$, and $\sigma(p_i)=p_{k+i}$, for $i\notin I$.
Then, $\varphi(p_1,\ldots,p_k) \not \der (\varphi(p_1,\ldots,p_k))^\sigma$.
\end{lemma}
\begin{proof}
By the assumption that $\varphi$ is very significant, we know that this term function is not a projection-conjunction. Thus, given that $I\subseteq\{1,\ldots,k\}$ is the exact set of indices such that $\varphi(p_1,\ldots,p_k)\der p_i$, for every $i\in I$, we conclude that $\{p_i:i\in I\}\not\der \varphi(p_1,\ldots,p_k)$.  There must be, then, some valuation~$v$ over $\{0,1\}$ such that $\val(p_i)=1$, for every $i\in I$, while $\val(\varphi(p_1,\ldots,p_k))=0$.
From the assumption about significance we also learn that $\varphi$ is not bottom-like, thus, in view of two-valuedness and the \Cref{twovalsignif}(ii), we know that there must be some valuation~$\val^\prime$ such that $\val^\prime(\varphi(p_1,\ldots,p_k))=1$.  Using the assumption that $\varphi(p_1,\ldots,p_k)\der p_i$ for every $i\in I$ one may conclude that $\val^\prime(p_i)=\val(p_i)=1$ for every $i\in I$.
Our final step to obtain a counter-model to witness $\varphi(p_1,\ldots,p_k) \not \der (\varphi(p_1,\ldots,p_k))^\sigma$ 
is to glue together the two latter valuations by considering a valuation~$\val^{\prime\prime}$ such that $\val^{\prime\prime}(p_j)=\val^{\prime}(p_j)$ for $1\leq j\leq k$ (satisfying thus the premise) 
and such that  $\val^{\prime\prime}(p_j)=\val(p_{j})$ for $j>k$ (allowing for the conclusion to be falsified).
\qed
\end{proof}

\begin{proposition}\label{prop:failingmaya}
The fibring $\fib{\cB_{\conn_1\!}}{\cB_{\conn_2}}$ of the logic of a very significant clas\-sical connective~$\conn_1$ and 
the logic of a non-top-like Boolean connective~$\conn_2$ distinct from~$\bot$
fails to be locally tabular, and therefore $\fib{\cB_{\conn_1\!}}{\cB_{\conn_2}}\subsetneq \cB_{\conn_1 \conn_2}$. 
\end{proposition}
\begin{proof}
{

We want to build over $\Sigma_1\cup\Sigma_2$, on a finite number of sentential variables, an infinite family $\{\conn_m\}_{m\in\mathbb{N}}$ of syntactically distinct formulas  
that are pairwise inequivalent according to $\fib{\cB_{\conn_1\!}}{\cB_{\conn_2}}$.

In case~$\conn_2$ is significant we know from \Cref{lem:TFalphas} that we can count on a singulary significant term function~$\psi_0$ allowed by $L_{\conn_2}(\{p\})\setminus P$.
Set, in this case, $\psi_{n+1}:=\psi_0\circ\psi_n$. 
Given the assumption that  $\cB_{\conn_2}$ is a two-valued logic, in view of \Cref{twovalsignif}(iii) it should be clear that no such~$\psi_{n+1}$ can be top-like.  To the same effect, in case $\conn_2$ is bottom-like, just consider any enumeration $\{\psi_m\}_{m\in\mathbb{N}}$ of the singulary term functions allowed by $L_{\conn_2}(\{p\})\setminus P$.
In both cases we see then how to build a family of syntactically distinct $\{\conn_2\}$-headed singulary term functions, and these will be used below to build a certain convenient family of ($\{\conn_1\}$-headed) formulas in the mixed language.}

In what follows we abbreviate $\conn_1(p_1,p_2,\ldots,p_{k_1})$ to~$C$.
We may assume, without loss of generality, that there is some $j<k_1$ such that $C\der_{\conn_1}p_i$ for every $i\leq j$ and $C\not\der_{\conn_1}p_i$ otherwise. 
Let $\sigma_n$, for each $n>0$, denote a substitution such that $\sigma_n(p_i)=p_i$, for $i\leq j$, and $\sigma_n(p_i)=\psi_{n\times i}(p)$ otherwise.
We claim that $C^{\sigma_a}\not\der_{\fib{\conn_1}{\conn_2}}C^{\sigma_b}$, for every $a\neq b$.


To check the claim, first note that, for each $a>0$, we have
$\St_{{\conn_1}{\fibS}{\conn_2}}(\{C^{\sigma_a}\})
=\{C^{\sigma_a}\}\cup \{p_i:i\leq j\}.$
From the fact that~$C$ is a significant term function, it follows that $\St_{{\conn_1}{\fibS}{\conn_2}}(\{C^{\sigma_a}\})$ is neither $\der_{\conn_1}$-explosive nor $\der_{\conn_2}$-explosive.
%
For arbitrary $b>0$, since $\Mmon_{\Sigma_2}(\psi_b(p))= \varnothing$, we have $\Mx^2_{{\conn_1}{\fibS}{\conn_2}}(\{C^{\sigma_a}\},\psi_b(p))=\varnothing$.
Therefore, using \Cref{bazooka}\ 
we may conclude that $C^{\sigma_a}\not\der_{{\conn_1}{\fibS}{\conn_2}}\psi_b(p)$ and, given that $\Mmon_{\Sigma_1}(C^{\sigma_b})\subseteq \{\psi_k(p):k\in \nats\}$,  
it also follows that $\Mx^1_{{\conn_1}{\fibS}{\conn_2}}(\{C^{\sigma_a}\},C^{\sigma_b})=\varnothing$.
Note, in addition, for each $n>0$, that $\skel_{\Sigma_1}(C)=C^{\sigma_n^\prime}$, where $\sigma_n^\prime(p_i):=p_i$ for $i\in I$, and $\sigma_n^\prime(p_i):=x_{\psi_{n\times i}}$ for $i\notin I$.
Therefore, given that $\conn_1$ is very significant, using \Cref{cor:seqmon} and \Cref{lem:projsubst} we conclude at last, for every $a\neq b$, 
that~$C^{\sigma_b}$ does not follow from~$C^{\sigma_a}$ according to $\fib{\cB_{\conn_1\!}}{\cB_{\conn_2}}$.  The latter combined logic, thus, fails to be locally tabular.
%
%
As a consequence, given that 
all two-valued logics are 
locally tabular 
we see that $\fib{\cB_{\conn_1\!}}{\cB_{\conn_2}}$ cannot coincide with ${\cB_{\conn_1\conn_2}}$.
  \qed
\end{proof}



\begin{example}
 If $\conn_1$ and $\conn_2$ are among the Boolean connectives mentioned in \Cref{theconnectives} then
 we have that $\fib{\cB_{\conn_1\!}}{\cB_{\conn_2}}\subsetneq \cB_{\conn_1\conn_2}$. 
\hfill$\triangle$
\end{example}

{ 

The following theorem makes use of the previous results to capture the exact circumstances in which the logic that merges the axiomatizations of two classical connectives coincides with the logic of these Boolean connectives.

\begin{theorem}\label{lem:colapse}
Consider the logic $\cB_{\conn_1\!}$ of the classical connective~$\conn_1$ and the logic $\cB_{\conn_2}$ of the distinct classical connective~$\conn_2$.  Then, $\fib{\cB_{\conn_1\!}}{\cB_{\conn_2}}=
\cB_{\conn_1\conn_2}$ 
iff  \underline{either}: 
\begin{itemize}
 \item[$\mathbf{(a)}$]  at least one among $\conn_1$ and $\conn_2$ is top-like, 
 \underline{or}
 \item[$\mathbf{(b)}$]  
 {%
 neither $\conn_1$ nor $\conn_2$ are very significant, 
 }%
 \underline{or}
 
 \item[$\mathbf{(c)}$] $\conn_1\in \cC_{\mathsf{BA}}^{\eq}$ and $\conn_2=\bot$ 
 {%
 (or $\conn_1=\bot$ and $\conn_2\in \cC_{\mathsf{BA}}^{\eq}$).
 }%
\end{itemize}
\end{theorem}

\begin{proof}
 The direction from right to left follows from \Cref{lem:clonestop,prop:projbot,lem:cloneseq}. 
 The other direction follows from \Cref{pro:connbot} and \Cref{prop:failingmaya}.
\qed\end{proof}

}

%
%

We can finally obtain the envisaged characterization result:

 {
\begin{theorem}


%

%
%

Let $\conw_1$ and $\conw_2$ be non-functionally complete disjoint sets of 
connectives such that  
 $\conw=\conw_1\cup\conw_2$ is functionally complete. 
 The disjoint fibring of the classical logics of $\conw_1$ and $\conw_2$ is classical
 iff
$\cC_{\mathsf{BA}}^{\conw_i}\in \{\mathbb{D},\mathbb{T}_0^\infty\}\cup \{\mathbb{T}_0^k:k\in \nats\} \mbox{ and }\cC_{\mathsf{BA}}^{\conw_j}= \mathbb{UP}_1$, {for some $i\in \{1,2\}$ and $j=3-i$}.
\end{theorem}
\begin{proof}
Note that if $\cC_{\mathsf{BA}}^{\conw_i}\in \{\mathbb{D},\mathbb{T}_0^\infty\}\cup \{\mathbb{T}_0^k:k\in \nats\}$  and $\cC_{\mathsf{BA}}^{\conw_j}= \mathbb{UP}_1$, for  $i\neq j\in \{1,2\}$, then we have that $\conw$ is functionally complete. 
For the right to left implication, it suffices to invoke Proposition~\ref{prop:conslifts} and item $\mathbf{(a)}$ of Theorem~\ref{lem:colapse}.

As for the converse implication, let us assume that $\cB_{{\conw}_1}\fibS \cB_{\conw_2}= \cB_{{\conw}}$. Using Proposition~\ref{prop:conslifts}, we know that for every pair of connectives $\conn_1\in\conw_1$ and $\conn_2\in\conw_2$ one of the items $\mathbf{(a)}$, $\mathbf{(b)}$ or $\mathbf{(c)}$ of Theorem~\ref{lem:colapse} must hold.
If $\mathbf{(a)}$ holds in all cases, then, without loss of generality, $\cC_{\mathsf{BA}}^{\conw_j}=\mathbb{UP}_1$. This, given the functional completeness of $\conw$, implies that $\cC_{\mathsf{BA}}^{\conw_i}\in  \{\mathbb{D},\mathbb{T}_0^\infty\}\cup \{\mathbb{T}_0^k:k\in \nats\}$.
Otherwise, we would have $\cC_{\mathsf{BA}}^{\conw_i}$ and $\cC_{\mathsf{BA}}^{\conw_j}$ both distinct from $\mathbb{UP}_1$, and items $\mathbf{(b)}$ or $\mathbf{(c)}$ of Theorem~\ref{lem:colapse} would have to hold in all the remaining cases. If $\mathbf{(b)}$ holds in all the remaining cases then we would conclude that $\conw_i\cup\conw_j$ contains only connectives that are not very significant, and that would contradict the functional completeness of $\conw$. Thence, without loss of generality, we could say that $\cC_{\mathsf{BA}}^{\conw_i}$ contains very significant connectives, and item $\mathbf{(c)}$ of Theorem~\ref{lem:colapse} would have to hold in those cases. But this would mean that $\cC_{\mathsf{BA}}^{\conw_i}\subseteq \cC_{\mathsf{BA}}^{\eq\land\top\bot}=\cC_{\mathsf{BA}}^{\eq\land\bot}$ and $\cC_{\mathsf{BA}}^{\conw_j}\subseteq \cC_{\mathsf{BA}}^{\top\bot}$. Note, however, that neither $\land$ nor $\bot$ can coexist in $\cC_{\mathsf{BA}}^{\conw_i}$ with $\eq$, or the underlying logic would express some very significant connective not expressible using only~$\eq$. We are therefore led to conclude that $\cC_{\mathsf{BA}}^{\conw_i}\subseteq \cC_{\mathsf{BA}}^{\eq}$ and $\cC_{\mathsf{BA}}^{\conw_j}\subseteq \cC_{\mathsf{BA}}^{\top\bot}$.  But this is impossible, as we would then have $\cC_{\mathsf{BA}}^{\conw}\subseteq\mathbb{A}$, contradicting the functional completeness of $\conw$.
    \qed
 \end{proof}
}

\section{Closing remarks}\label{sec:final}

\noindent
In the present paper, we have investigated and fully characterized the situations when merging two disjoint fragments of classical logic still results in a fragment of classical logic. 
As a by-product, we showed that recovering full classical logic in such a manner can only be done 
when one of the logics is a fragment of classical logic consisting exclusively of top-like connectives, while the other forms a functionally complete set of connectives with the addition of~$\top$. Our results take full advantage of the characterization of Post's lattice, and may be seen as an application of recent developments concerning fibred logics. Though our conclusions cannot be seen as a total surprise, we are not aware of any other result of this kind. Some unexpected situations do pop up, like the fact that $\fib{\cB_{\eq}}{\cB_{\bot}}=\cB_{\eq\bot}$, or the fact that $\fib{\cB_{\,\coimplsub}}{\cB_{\top}}$ and $\fib{\cB_{\lor +}}{\cB_{\top}}$ both yield full classical logic. The latter two combinations are particularly enlightening, given that according to~\cite{deccomp} the complexity of disjoint fibring is only polynomially worse than the complexity of the component logics, and we know from~\cite{Beyersdorff} 
that the decision problems for 
$\cB_{\,\coimplsub}$ or  $\cB_{\lor +}$ are both $\mathbf{co}$-$\mathbf{NP\text{-}complete}$,
as in full classical logic.
{As a matter of fact, some of the results we obtained may alternatively be established as consequences of the complexity result in~\cite{deccomp}
together with the conjecture that $\mathbf{P}\neq \mathbf{NP}$. In fact, for disjoint sets of Boolean connectives $\conw_1$ and $\conw_2$ such that $\conw_1\cup \conw_2$ is functionally complete, if the decision problems for $\cB_{\conw_1}$ and for $\cB_{\conw_2}$ are both in $\mathbf{P}$ then clearly $\cB_{\conw_1}\fibS \cB_{\conw_2}\neq \cB_{\conw_1\cup \conw_2}$.}
However, the techniques we use here do not depend on $\mathbf{P}\neq \mathbf{NP}$ and
allow us to solve also the cases 
 in which the complexity of the components is already in $\mathbf{co}$-$\mathbf{NP}$, for which the complexity result in~\cite{deccomp} offers no hints.

Similar studies could certainly be pursued concerning logics other than classical.
However, even for the classical case there are some thought-provoking unsettled questions. Concretely, we would like to devise semantical counterparts for all the combinations that do not yield fragments of classical logic, namely those covered by \Cref{prop:failingmaya}. So far, we can be sure that such semantic counterparts cannot be provided by a single finite logical matrix. Additionally, we would like to link the cases yielding fragments of classical logic (as covered by the conditions listed in Theorem~\ref{lem:colapse}) to properties of the multiple-conclusion consequence relations~\cite{SS:MCL:78} pertaining to such connectives.

\bibliographystyle{plain} 
\bibliography{biblio.bib}
\end{document}